\documentclass[
final
]{dmtcs-episciences}

\usepackage[utf8]{inputenc}
\usepackage{subfigure}
\usepackage[round]{natbib}
\usepackage{graphicx,amssymb,amsmath,amsthm}
\usepackage{url,float}
\usepackage{color}

\usepackage{xcolor}
\usepackage{hyperref}
\hypersetup{
allcolors=black
}

\newcommand{\hide}[1]{}

\newtheorem{observation}{Observation}
\newtheorem{claim}{Claim}
\newtheorem{theorem}{{\bf Theorem}}
\newtheorem{lemma}{Lemma}
\newcommand{\lemlab}[1]{\label{lemma:#1}}
\newcommand{\thmlab}[1]{\label{thm:#1}}
\newcommand{\figlab}[1]{\label{fig:#1}}
\newcommand{\seclab}[1]{\label{sec:#1}}
\newcommand{\lemref}[1]{\ref{lemma:#1}}
\newcommand{\thmref}[1]{\ref{thm:#1}}

\newcommand{\secref}[1]{\ref{sec:#1}}
\newcommand{\figref}[1]{\ref{fig:#1}}

\def\defn#1{\textit{\textbf{\boldmath #1}}}

\newcommand{\squeezelist}{\setlength{\itemsep}{0pt}}

\usepackage{appendix}
\usepackage[shortlabels]{enumitem}

\publicationdata{vol. 28:3}{2026}{22}{10.46298/dmtcs.15007}{2024-12-29; 2024-12-29; 2026-11-09}{2026-11-10}

\begin{document}

\author[MIT CompGeom Group]{MIT CompGeom Group%
\affiliationmark{1}%
\thanks{Artificial first author to highlight that the other authors
(in alphabetical order) worked as an equal group.}
  \and{Hugo A. Akitaya}%
  \affiliationmark{2}
  \and{Erik D. Demaine}%
  \affiliationmark{3}
  \\
  \and{Adam Hesterberg}%
  \affiliationmark{4}
  \and{Anna Lubiw}%
  \affiliationmark{5}
  \and{Jayson Lynch}%
  \affiliationmark{6}
  \\
  \and{Joseph O'Rourke}%
  \affiliationmark{7}
  \and{Frederick Stock}%
  \affiliationmark{8}
 }
  
\title{Super Guarding and Dark Rays in Art Galleries}

\affiliation{
 MIT, Boston, USA\\
 U. Mass. Lowell, Lowell, USA\\
 MIT, Boston, USA\\
 Harvard Univ., Boston, USA\\
 Univ. Waterloo, Waterloo, Canada\\
 MIT, Boston, USA\\
 Smith College, Northampton, USA\\
 U. Mass. Lowell, Lowell, USA
 }

\keywords{Art Gallery Theorem, visibility, simple polygon, convex polygon}

\maketitle

\noindent\vspace*{-7.51pt}
\begin{abstract}
We explore an Art Gallery variant where each point of a polygon
must be seen by $k$ guards, and guards cannot see through other guards.
Surprisingly, even covering convex polygons under this variant is not
straightforward. For example, covering every point in a triangle
$k\hspace{1pt}{=}\hspace{1pt}4$ times (a \defn{$4$-cover})
requires $5$ guards, and achieving
a $10$-cover requires $12$ guards.
Our main result is tight bounds on $k$-covering a convex polygon of $n$ vertices,
for all $k$ and $n$.
The proofs of both upper and lower bounds are nontrivial.
We also obtain bounds for simple polygons, 
leaving tight bounds an open problem.
\end{abstract}

\section{Introduction}
\seclab{Introduction1}
The original Art Gallery Theorem showed
that $\lfloor n/3 \rfloor$ guards are sometimes
necessary and always sufficient to guard 
a simple polygon $P$ of $n$ vertices; see~\cite{o-agta-87}. 
(Throughout, $P$ includes its boundary $\partial P$,
and guarding $P$ includes guarding $\partial P$.)
There have been many interesting variants
explored since then.
In this paper we explore two variants that are interesting in combination, although not individually.
\begin{enumerate}[label={\rm(\arabic*)}]
\item \emph{Guards blocking guards}:
Suppose guards cannot see through other guards.\footnote{
This was posed as an exercise in
\cite{do-dcg-11}, Exercise~1.28, p.~14.} 
More precisely, if $g_1$ and $g_2$ are guards, and $g_1, g_2, p$ are on a line in that order, 
then point $p$ is not visible from $g_1$. 
Still the original bound $\lfloor n/3 \rfloor$ holds, because $g_2$
can continue $g_1$'s line-of-sight to $p$, picking it up where that line-of-sight terminates at $g_2$.
\item \emph{Multiple coverage}:
Suppose every point in the 
closed polygon must be seen by $k$ guards
i.e., the guards must \defn{$k$-cover} the polygon. 
The problem of $k$-guarding
has been explored under various restrictions on guard location: 
e.g., \cite{belleville1994k,salleh2009k,busto2013k,durocher2}.
If multiple guards can be co-located at the same point, then this is trivial.
If co-location is disallowed, but guards can
see through other guards, then this still reduces to the case $k=1$ 
since we can replace a single guard by a cluster of $k$ guards.
(We detail the argument in Section~\secref{SimplePolygon}.)
\end{enumerate}

So neither of these variations is ``interesting" by itself in the sense that easy arguments lead
to $\lfloor n/3 \rfloor$ bounds.
However, consider now mixing these two variants:
\begin{quote}
\emph{Q}:
How many guards are necessary and sufficient to 
cover a simple polygon $P$ of $n$ vertices
so that every point of $P$ is seen by at least $k$ guards,
where guards cannot be co-located, 
and each guard blocks lines-of-sight 
beyond it.
\end{quote}
To our surprise, answering \emph{Q} is not straightforward, even for convex polygons,
even for triangles.
For example, to cover a triangle to depth $k=3$, one guard at each vertex suffices.
Note 
we consider a guard to see itself.
But to cover to depth $k=4$ requires $g=5$ guards;
see Fig.~\figref{4Cover5Guards}.
And covering to depth $k=10$ requires $g=12$ guards.\footnote{
Note that if we only cared to $k$-cover
points strictly interior to $P$, it would suffice
to place $k$ guards on $\partial P$.}

\begin{figure}[htbp]
\centering
\includegraphics[width=0.8\textwidth]{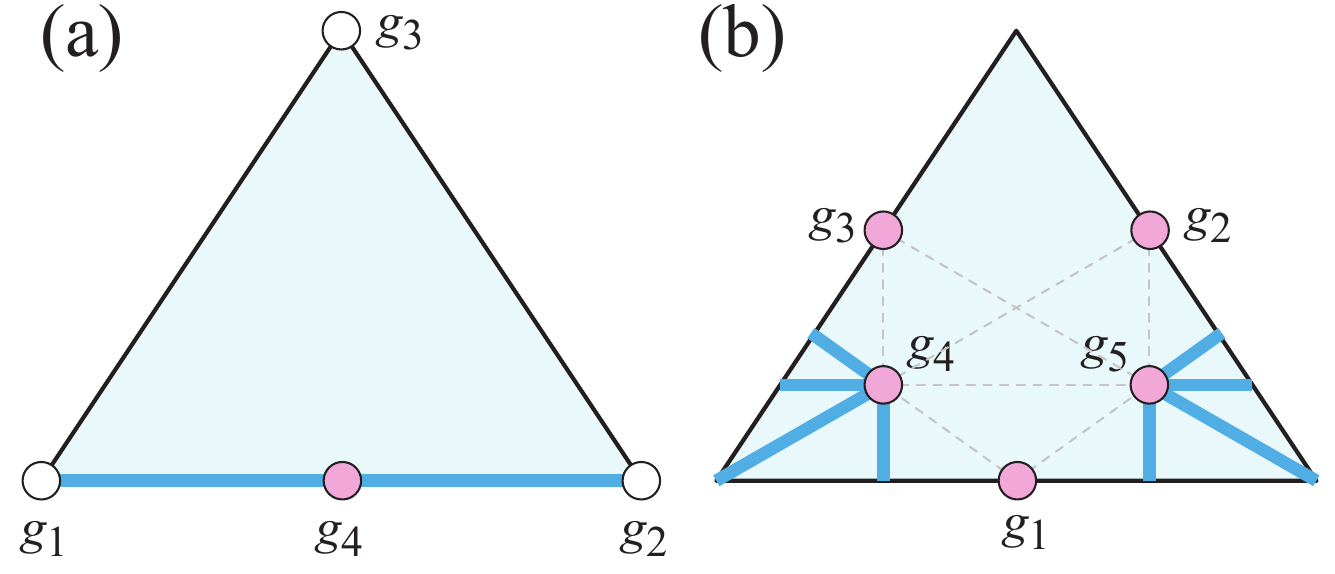}
\caption{Five guards needed to $4$-cover.
(a)~All strictly interior points are $4$-covered, but the blue segments
to either side of $g_4$ are only $3$-covered.
(b)~Points on the dark rays (blue segments)
incident to $g_4$ and $g_5$ are $4$-covered; all
other points are $5$-covered.}
\figlab{4Cover5Guards}
\end{figure}

The main result of this paper is the following theorem,
which holds for any shape of convex polygon.
We use $n$ for the number of vertices, $k$ for the depth of cover, and $g$ for the number of guards.

\begin{theorem}
\thmlab{ConvexPoly}
\label{thm:main}
For a closed convex $n$-gon,
coverage to depth $k$ requires $g \in \{k,k+1,k+2\}$ guards:
\begin{enumerate}[label={\rm(\arabic*)}]
\squeezelist
\item For $k \le n$: $\;g=k$ guards are necessary and sufficient.
\item For $n < k < 4n-2$: $\;g=k+1$ guards are necessary and sufficient.
\item For $4n-2 \le k$: $\;g=k+2$ guards are necessary and sufficient.
\end{enumerate}
\end{theorem}
Thus there are three regimes depending on the relationship between $n$ and $k$.
For triangles, $n=3$, the following table details those regimes:

\begin{table}[h]
\begin{center} 
\begin{tabular}{| c || c | c | c || c | c | c | c | c | c || c | c | c |}
\hline
$k$ & 1 & 2 & 3 & 4 & 5 & 6 & 7 & 8 & 9 & 10 & 11 & $\cdots$\\
\hline
$g$ & 1 & 2 & 3 & 5 & 6 & 7 & 8 & 9 & 10 & 12 & 13 & $\cdots$\\
\hline
\end{tabular}
\end{center}
\end{table}

\noindent
Another example: For $n=4$, $g=14$ guards $13$-cover, but a $14$-cover requires $g=16$ guards.
See ahead to Fig.~\figref{Square_FS}.

Our primary focus is proving Theorem~\thmref{ConvexPoly}.
We also obtain in Lemma~\lemref{Wedge} tight bounds for a convex wedge,
which can be viewed as a $2$-sided 
unbounded convex polygon.
Finally, we briefly address simple polygons in Theorem~\thmref{SimplePolygon},
which we do not consider as natural a fit as is the question for convex polygons.\footnote{Preliminary version: \cite{MIT-DarkRays-2023}.}

\medskip
\subsection{Proving Theorem~\ref{thm:main} using Dark Rays and Dark Points}
\seclab{DarkRays}
With some abuse of notation, we will identify both a guard and that guard's location as $g_i$.
Let $g_1$ and $g_2$ be two guards visible to one another.
We say that $g_2$ \defn{generates} a \defn{dark ray} \defn{at} $g_1$, 
which is
a ray 
contained in the line through $g_1$ and $g_2$,
incident to and rooted but open at $g_1$, and invisible to $g_2$.
And similarly, $g_1$ generates a dark ray at $g_2$.

A point is called 
\defn{dark} if it is contained in a dark ray, and 
\defn{$d$-dark} if 
it is contained in at least $d$ dark rays.

Because 
a $d$-dark point is hidden from $d$ guards,
we obtain an immediate relationship between dark rays 
and multiple guarding for convex polygons.

\begin{observation}
\label{obs:guards-and-rays-1}
$k$-guarding with $g=k$ guards is possible 
if and only if
there is no dark point inside $P$, i.e., all dark rays are strictly exterior to $P$.
\end{observation}

\begin{observation}
\label{obs:guards-and-rays-2}
$k$-guarding with $g=k+1$ guards is possible 
if and only if
there is no $2$-dark point inside $P$.
\end{observation}

\begin{observation}
\label{obs:guards-and-rays-3}
$k$-guarding with $g=k+2$ guards is possible if and only if there is no $3$-dark point inside $P$. 
Furthermore, this is always possible because we can perturb the guards to avoid $3$-dark points, as proved formally in Lemma~\ref{lemma:avoid-3-dark} below.\end{observation}

We begin by proving the claim in Observation~\ref{obs:guards-and-rays-3} about perturbing guards to avoid $3$-dark points. 
For use in later sections
we extend the result to handle a \defn{wedge} or \defn{cone}---a region of the plane
bounded by two rays from a convex vertex $a$.

\begin{lemma}[Guards in General Position]
\label{lemma:avoid-3-dark}
For any $g$ it is possible to place $g$ guards in any convex $n$-gon $P$
or any wedge
without creating a 3-dark point in $P$.  
\end{lemma}

\begin{proof}
Although this follows from general perturbation results, we give a straightforward inductive construction.

We show how to place $g$ guards in a specified open region of the plane (a convex polygon 
or near the vertex of 
a wedge)
while avoiding $3$-dark points anywhere in the plane.

Place the guards sequentially.
After placing $i$ guards,
let $\mathcal{A}_i$ be the arrangement of lines determined by:
(a)~pairs of guard points; 
and (b)~a guard point and a $2$-dark point at the intersection of two dark rays.
(For $i \le 3$ noncollinear guards, there are no $2$-dark points.)
Place the $(i+1)$-st guard at any point in 
the open region, not on a line of $\mathcal{A}_i$.
This is possible since the region is open.
Note that this avoids three collinear guards and the crossing of three dark rays. Now update the arrangement to $\mathcal{A}_{i+1}$ and repeat.
\end{proof}

Observations~\ref{obs:guards-and-rays-1},~\ref{obs:guards-and-rays-2}, and~\ref{obs:guards-and-rays-3} recast the question of how many guards are needed to $k$-guard a polygon, 
to
the question of how many guards can be packed into the polygon without creating $1$-dark, $2$-dark or $3$-dark points.
Consider the boundary between the first and second regimes of Theorem~\ref{thm:main}: when is $k$-guarding with $k$ guards possible?  
Equivalently, what is the maximum number of guards that can be placed without creating dark points? 
We can place up to $n$ guards at the polygon vertices without creating dark points.  
To prove that $n$ is a tight bound, we use the following lemma: 

\begin{lemma}
\label{lemma:create-1-dark}
Any placement of more than $n$ guards in a convex $n$-gon $P$ results in a dark point in~$P$.    
\end{lemma}
\begin{proof} 
If a guard $g_0$ is strictly internal to $P$, then there is a dark ray at $g_0$ generated by every
other guard. So it must be that all guards are on $\partial P$.

View each edge of $P$ as half-open, including its clockwise endpoint but not its counterclockwise endpoint.
So the edges are disjoint and their union is $\partial P$.
Every edge $e$ can contain at most one guard: If $e$ contains two or more guards, one of them, say $g_1$, is interior to $e$
and so there is a dark ray at $g_1$ along $e$.
So there can be at most $n$ guards while avoiding dark points.
\end{proof}

The boundary between the second and third regimes of Theorem~\ref{thm:main}---the
$4n-2$ threshold for $k$-guarding with $k+1$ guards---is established by the following:

\begin{theorem}
\label{thm:no-2-dark}
The maximum number of guards that can be placed in a convex $n$-gon without creating $2$-dark points is $4n-2$.
Equivalently:
\begin{itemize}
\item[(A)] 
\label{part:create-2-dark}
Any placement of more than $4n-2$ guards in a convex $n$-gon $P$ results in a $2$-dark point in~$P$.
\item[(B)] 
\label{part:avoid-2-dark}
It is possible to place $4n-2$ guards in any convex $n$-gon $P$ without creating a 2-dark point in~$P$.  
\end{itemize}
\end{theorem}

We prove part~(A) 
in Section~\secref{UpperBound}
and we prove part~(B)
by a direct construction  
in Section~\secref{LowerBound}.
Both parts
are non-trivial, 
and their proofs constitute the main focus of the paper.
Assuming these results, the proof of Theorem~\thmref{ConvexPoly}  proceeds as follows:

\begin{proof}[of Theorem~\ref{thm:main}]
\ 
\begin{enumerate}
\item For $k \le n$, $g=k$.  Clearly, $k$ guards are necessary.  For sufficiency, place $k$ guards at vertices of polygon $P$.  Then all dark rays are exterior to $P$, so this is a $k$-cover (Observation~\ref{obs:guards-and-rays-1}).

\item For $n < k < 4n-2$, $g=k+1$. 
For necessity, we argue that $k$ guards are not enough.
By Lemma~\ref{lemma:create-1-dark}, since $k>n$, there will always be a dark point in $P$. So the guards do not $k$-cover the polygon (Observation~\ref{obs:guards-and-rays-1}).

For sufficiency, we have $g = k+1 \le 4n-2$ so by  Theorem~\ref{thm:no-2-dark}(B) we can place $g$ guards without creating a $2$-dark point, thus providing a $k$-cover (Observation~\ref{obs:guards-and-rays-2}).

\item For $4n-2 \le k$, $g=k+2$.  For necessity, $k+1$ is greater than $4n-2$ so by  Theorem~\ref{thm:no-2-dark}(A) any placement of $k+1$ guards creates a $2$-dark point, so the guards do not $k$-cover the polygon (Observation~\ref{obs:guards-and-rays-2}).

For sufficiency, Lemma~\ref{lemma:avoid-3-dark} shows that we can place $k+2$ guards without creating a $3$-dark point.  These guards provide a $k$-cover (Observation~\ref{obs:guards-and-rays-3}). 
\end{enumerate}
\end{proof}

\section{Theorem~\ref{thm:no-2-dark}(A): At Most \boldmath $4n-2$ Guards}
\seclab{UpperBound}
In this section we prove that at most $4n-2$ guards can be placed in a convex $n$-gon $P$ without creating $2$-dark points in $P$.

\subsection{Triangle Lemma}
\seclab{TriangleLemma}
The following lemma is a key tool in the proof of the upper bound.
It establishes that, excluding an
exceptional case, any triangle of guards in $P$ may only
contain one additional guard if we are to avoid $2$-dark points in $T$.

\begin{lemma}[Triangle]
\lemlab{GuardTriangle}
Suppose some guards are placed in $P$ without creating $2$-dark points. Let $T$ be a closed triangle in $P$ with guards $g_1,g_2,g_3$ at its corners.
Then, with one exception, $T$ contains at most one more guard.  

The exceptional case 
allows two guards, $g_4,g_5$, in $T$
when
(up to relabeling)
$g_1 g_3$ is an edge of 
polygon $P$, 
$g_4$ lies on that edge, and
$g_5$ lies on segment $g_2 g_4$.
\end{lemma}
\begin{proof}
Refer to Fig.~\figref{AnnaPfTri2}(a,b) throughout.
We first prove that there cannot be two guards strictly interior to $T$.
Suppose $g_4$ is strictly interior to $T$.
Then $g_1,g_2,g_3$ generate three dark rays at
$g_4$, each of which crosses a different edge of $T$.
The same would be true for a second strictly interior guard $g_5$. 
So a dark ray 
at $g_5$ must cross a dark ray at $g_4$ to reach an edge of $T$. 
The result is a $2$-dark point, marked $x$ in (a) of the figure.
Since we assumed no $2$-dark points in $P$, there cannot be two extra guards interior to $T$.

To prove the lemma,  suppose there is  more than one extra guard.  We establish the conditions for the exceptional case.  By the above argument, a guard must lie on an edge of $T$.
Suppose 
that $g_4$ lies on edge $e=g_1 g_3$ of $T$.
Then left and right of $g_4$ on $e$ are dark rays 
generated by $g_1$ and $g_3$.
Placing another guard $g_5$ at any point not 
on segment $g_2 g_4$
leads to a dark ray at $g_5$, 
generated by $g_2$, crossing $e$ to form a $2$-dark point there. 
Since there is at most one guard strictly internal to $T$, there cannot be yet another guard $g_6$ also on segment $g_2 g_4$.

We are left with the situation illustrated in (b) of the figure, where there are two extra guards: $g_4$ lying on edge $g_1 g_3$ of $T$ and $g_5$ lying on segment $g_2 g_4$.  There are no $2$-dark points inside $T$.
It remains to prove that $g_1 g_3$ is actually an edge of polygon $P$.

The  dark ray at $g_5$ generated by $g_2$  contains the dark ray at $g_4$ generated by $g_5$ so, to avoid $2$-dark points inside $P$, $g_4$ must be on the boundary of $P$.  By the same argument, $g_1$ and $g_3$ must be vertices of $P$.
\end{proof}
\begin{figure}[htbp]
\centering
\includegraphics[width=0.8\textwidth]{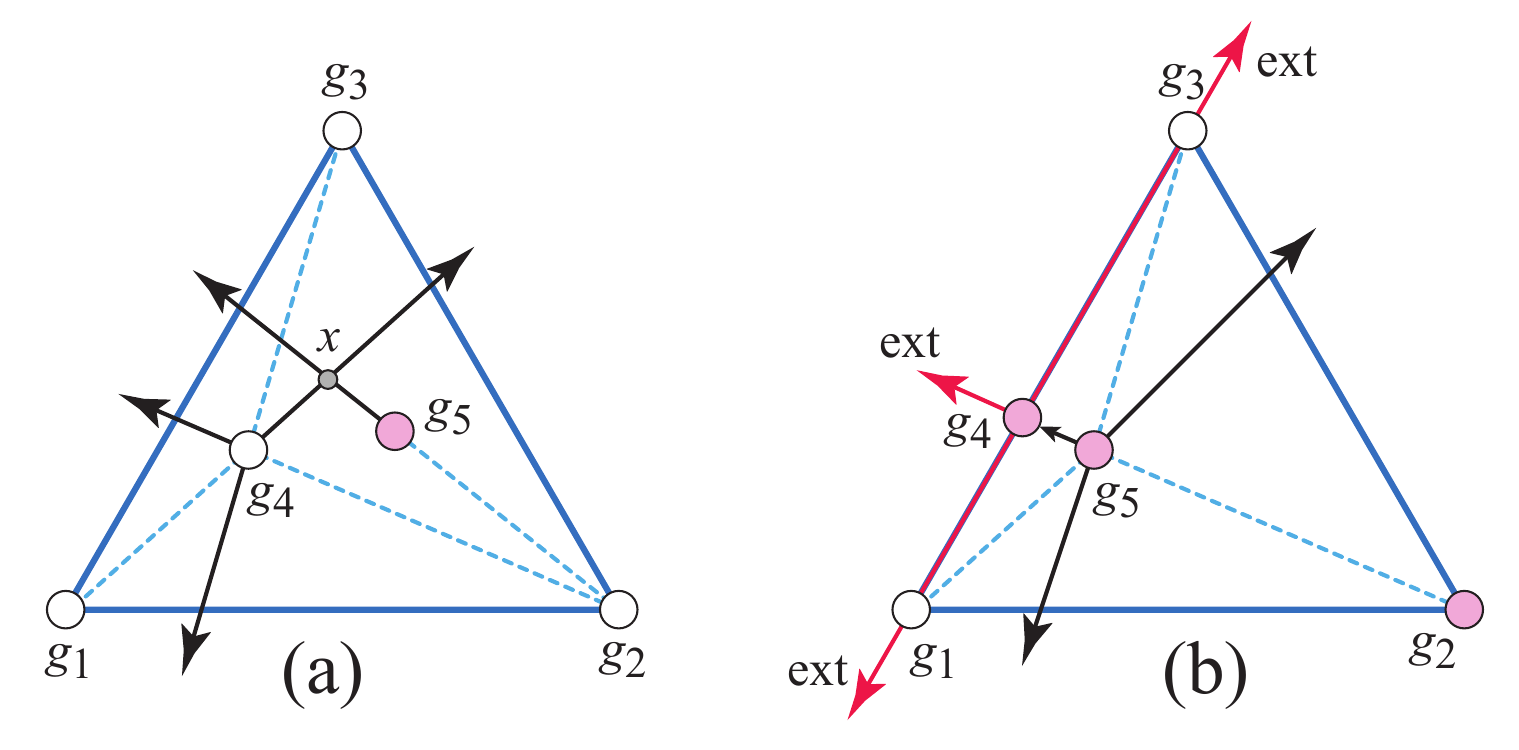}
\caption{
In this and following figures, guards are indicated by hollow circles.
(a)~Generic placements of $g_4,g_5$ produce a $2$-dark point $x$.
(b)~The exceptional case, with dark rays exterior to~$P$.
}
\figlab{AnnaPfTri2}
\end{figure}

\medskip
We now turn to the $4n-2$ bound.
Consider a placement of guards in $P$ such that there are no $2$-dark points in $P$. Our goal is to prove that there are at most $4n-2$ guards.
Let \defn{$C$} be the convex hull of the guards.
The main idea is as follows.
We will show in Lemma~\lemref{Chull} that 
the number of guards on $\partial C$, not counting collinear 
guards interior to $P$, is at most $2n$.
Triangulating $C$ leads to at most $2n-2$ triangles.
Lemma~\lemref{GuardTriangle} then shows that there is at most one extra guard 
inside each triangle,
which leads to the $4n-2$ upper bound.
To make this rigorous, we must take into account collinear guards and the exceptional case of Lemma~\lemref{GuardTriangle}.

First shrink $P$ 
so that it maximally touches $C$, as follows.
Move each edge of $P$ parallel to itself 
toward the interior until it hits a guard.
If an edge $e$ only has a guard at one endpoint, then rotate $e$ about that endpoint toward
the interior until it hits another guard.
The 
reduced polygon contains all the guards, has no $2$-dark point, and has at most $n$ vertices, so it suffices to prove the bound on the number of guards for the 
reduced polygon.
Henceforth we 
adopt the \emph{Shrunken Polygon Assumption}:
every edge of $P$ has either one or more guards in its interior, or a guard at 
both endpoints.

The proof requires careful handling of collinear guards:
a guard is called \defn{collinear}
if it lies on a line between two other guards.

Define \defn{$G^*$} 
to consist of the guards on $\partial P$ together with any guard
that is a corner of $C$ in the interior of $P$.
So collinear guards on $\partial P$ are in $G^*$, but collinear guards
on $\partial C$ and internal to $P$ are excluded from $G^*$.
See Fig.~\figref{gstar}.
Define \defn{$g^*$}$=|G^*|$.
This is the key count that is needed to complete the upper-bound proof.
\begin{figure}[htbp]
\centering
\includegraphics[width=0.8\textwidth]{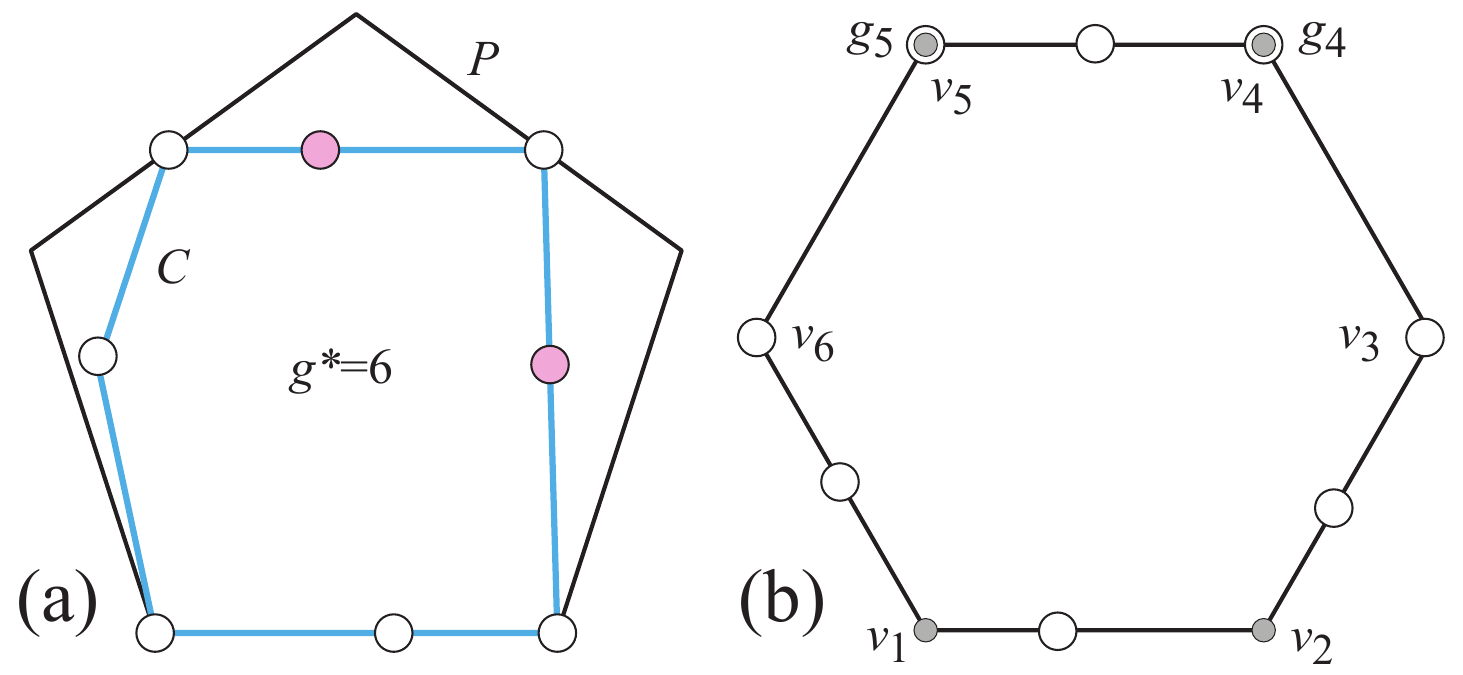}
\caption{(a)~The two pink guards are not included in 
$G^*$.
(b)~$v_1,v_2$ are 
darkened 
but have no guard; $g_4,g_5$ are both guards and 
darkened vertices.
So $d=4$ and 
$g^P = n + \frac{1}{2} d = 8$.}
\figlab{gstar}
\end{figure}

\begin{lemma}
\lemlab{Chull}
The number of guards $g^*$ as defined above is at most $2n$.
\end{lemma}
\begin{proof}
Let $g^P$ be the number of guards on $\partial P$ and let $c$ be the number of guards that are corners of $C$ in the interior of $P$.
As noted above, $g^* = g^P + c$.  We will bound $g^P$ and $c$ separately.  Both bounds are in terms of the number of darkened vertices, where a vertex $v$ of $P$ is \defn{darkened} if guards on $\partial P$
generate a dark ray through $v$.

We first bound $g^P$.
The main observation we will use to limit $g^P$
is that a vertex $v$ cannot be darkened from both
incident edges, as that would 
render $v$ a $2$-dark point.
The idea is to count guards and darkened vertices per edge.  A guard internal to an edge counts towards the edge, and a vertex guard counts half towards each incident edge.  More precisely, 
for an edge $e$, let $g(e)$ be the number of guards internal to $e$ plus half the number of vertex guards on $e$.  Then $g^P = \sum_e g(e)$.  

Fig.~\figref{HalfCounts} shows the possibilities: $g(e) =2$,
either from two internal guards, or one internal guard and two endpoint guards; 
$g(e) = 3/2$ 
from one endpoint guard and one internal guard; or $g(e)=1$ from one internal guard or two endpoint guards.

These are the only possibilities:
(a)~By the Shrunken Polygon Assumption, every edge has at least one guard,
and if it has only one guard, the guard must be internal to the edge.
(b)~All possibilities for two guards on an edge are included. 
(c)~An edge can only have three guards 
when two are at the endpoints of the edge: an endpoint without a guard would be rendered $2$-dark by the three guards on the edge.
(d)~An edge cannot have four or more guards, as then the extreme points would be 
at least $2$-dark.

Next we observe from Fig.~\figref{HalfCounts} a relationship between $g(e)$ and \defn{$d(e)$},
which is defined to be the number of dark rays on edge $e$ generated by guards on $e$: 
if $g(e)=2$ then $d(e)=2$;
if $g(e)=1 \frac{1}{2}$ then $ d(e)=1$;
and if $g(e)=1$ then $ d(e)=0$.  Equivalently, 
$d(e) = 2(g(e) - 1)$. 

Finally, we note that \defn{$d$}, defined to be the number of darkened vertices, is $\sum_e d(e)$, since each dark ray on $e$ darkens an endpoint of $e$, and no vertex can be darkened from both incident edges.

Putting these together,
$$d = \sum_e d(e) = \sum_e 2(g(e)-1) = 2 \sum_e g(e) -2n = 2 g^P -2n$$
which gives 

\begin{equation}
\label{eqn:g-sup-P}
g^P = n + \tfrac{1}{2} d \ .
\end{equation}

For example, for even $n$,
placing a guard at every vertex and a guard in the interior of every other edge
darkens every vertex, so $g^P = \frac{3}{2} n$.

\begin{figure}[htbp]
\centering
\includegraphics[width=0.8\textwidth]{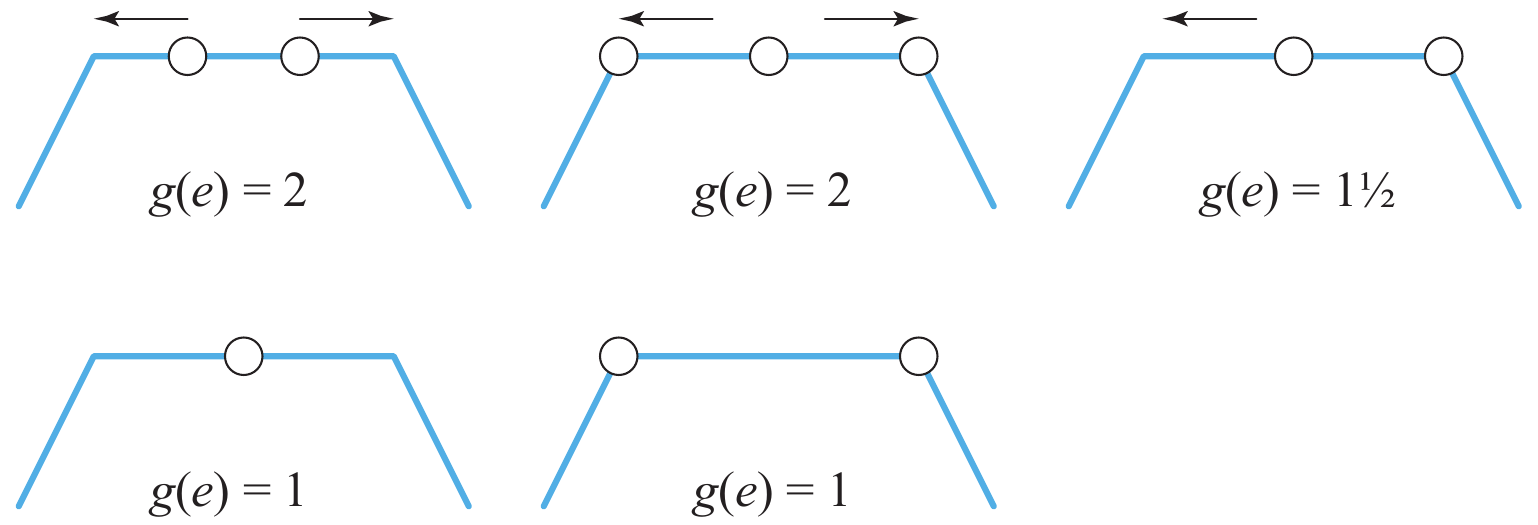}
\caption{Edge counts. Arrows indicate darkened vertices.}
\figlab{HalfCounts}
\end{figure}

We next bound $c$,  
the number of guards strictly 
interior to $P$ 
that are corners of $C$.
Let $g_0$ be such a corner guard.
Moving 
clockwise and counterclockwise
on $C$, let $g_1$ and $g_2$ be the first guards that are on $\partial P$, say on edges $e_1$ and $e_2$.
(Recall that every edge of $P$ has a guard, and so $g_1$ and $g_2$ exist and are distinct.)
Note that there cannot be another vertex of $C$ internal to $P$ between $g_1$ and $g_2$,
as then two dark rays would cross inside $P$:
see Fig.~\figref{C_corners}(a).
Also note that $g_0$ is not collinear with $g_1$ and $g_2$, because we
are counting $g^*$, which excludes collinear guards on $C$.
Since every edge has a guard, edges $e_1$ and $e_2$ must be incident at a vertex $v$ of $P$, and $v$ has no guard (because otherwise $g_0$ would be internal to $C$).  The dark rays incident to $g_0$ from $g_1$ and $g_2$ cross $e_1$ and $e_2$ as shown in Fig.~\figref{C_corners}(b).  So $v$ cannot be darkened by the guards on $e_1$ or $e_2$ otherwise again two dark rays would cross. 

Thus each guard $g_0$ counted in $c$ 
corresponds to a non-darkened vertex,
so $c \le n-d$.

Combining with Equation~\ref{eqn:g-sup-P},
$$g^* = g^P + c \le n + \tfrac{1}{2}d + (n-d) = 2n - \tfrac{1}{2}d \le 2n \;.$$

Equality is achieved when there is one guard internal to each edge, and one guard inside $P$ between each consecutive pair, 
and no collinear guards nor darkened vertices of $P$. See Fig.~\figref{C_corners}(c).
\end{proof}

\begin{figure}[htb]
\centering
\includegraphics[width=0.6\textwidth]{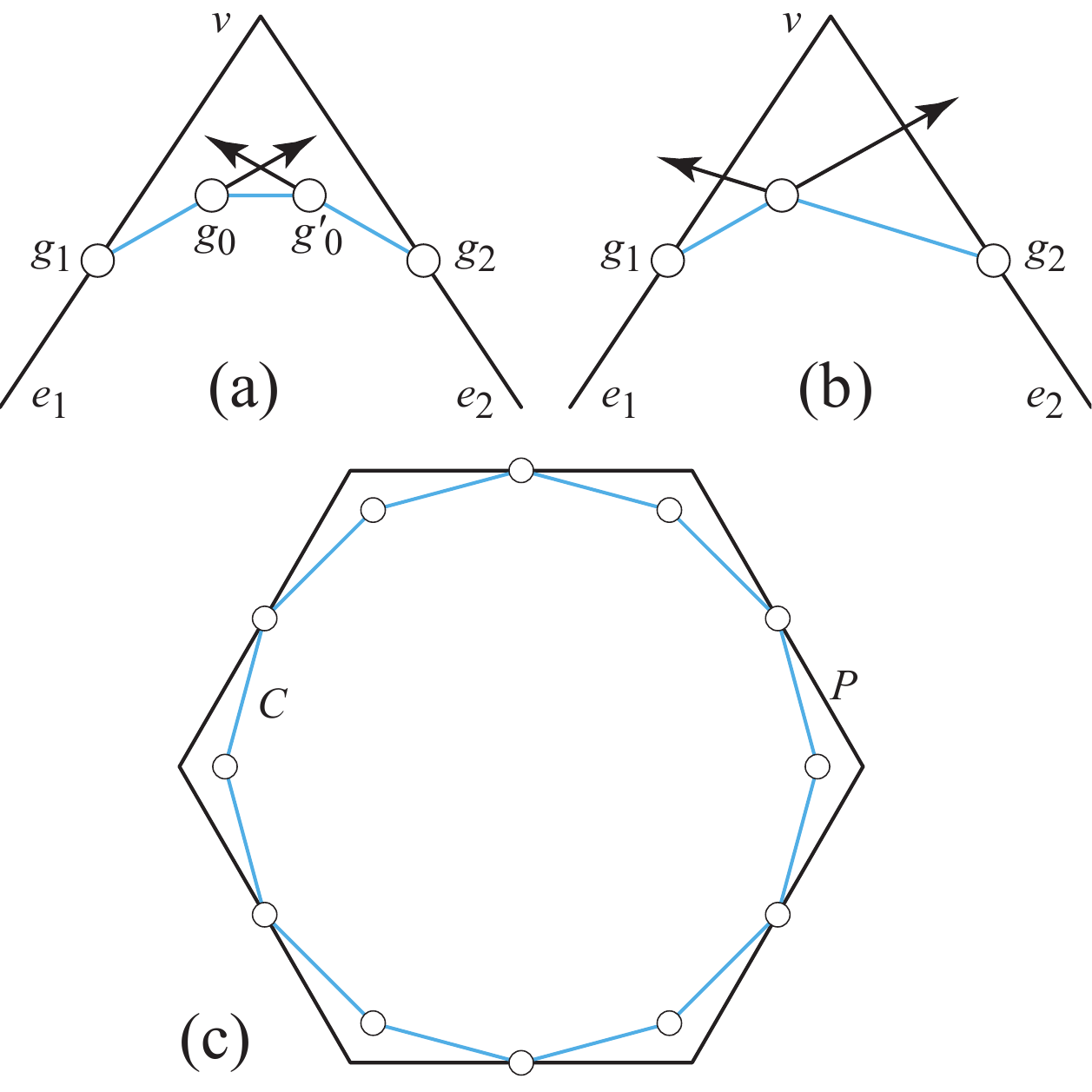}
\caption{(a)~$g_0$ and $g'_0$ create intersecting dark rays in $P$.
(b)~$v$ cannot be a darkened vertex.
(c)~The upper bound $g^*=2n$ can be achieved.}
\figlab{C_corners}
\end{figure}

We are now ready to prove the main result of this section.

\begin{proof}[Theorem~\ref{thm:no-2-dark}(A)]
Consider a placement of guards inside $P$ that avoids $2$-dark points. 
We prove that the number of guards is at most $4n-2$.
We use $G^*$ and $g^*$ as defined above.
By Lemma~\lemref{Chull},
$g^* \le 2n$. Triangulate the guards in $G^*$, 
i.e., add a maximal set of non-crossing chords between pairs of guards in $G^*$.
Recall that $G^*$
includes collinear guards on $\partial P$ but excludes collinear guards internal to $P$.

There are at most $2n-2$ triangles in this triangulation.  
We will apply Lemma~\lemref{GuardTriangle}. Note that  the exceptional case of the lemma only happens when 
a guard $g$ inside the triangle lies on $\partial P$,
but all the guards on $\partial P$ were already included in $G^*$ so the exceptional case cannot occur
because there is a triangulation edge incident to $g$.
Thus, by
Lemma~\lemref{GuardTriangle}, there is at most one extra guard in each triangle, 
for a total of at most $2n + (2n-2) = 4n-2$ guards.
\end{proof}

\section{Theorem~\ref{thm:no-2-dark}(B): Placing \boldmath $4n-2$ Guards}
\seclab{LowerBound}
The challenge is to 
locate $g=4n-2$ guards so that 
there are no $2$-dark points in $P$, i.e., so that no two dark rays intersect inside the polygon.

We first show in Section~\ref{sec:4n-2achievable} how to place $g=10$ guards in a triangle and $g=14$ guards in a square 
without dark rays intersecting, while hinting at the general strategy.
Section~\secref{GeneralLowerBound} contains the general proof.

\subsection{\boldmath $g=4n-2$ guards achievable for triangle and square}
\seclab{4n-2achievable}
Fig.~\figref{FreddyTriangle} illustrates a placement of $10$ guards
in a triangle $P$
such that all dark-ray intersections are strictly exterior to $P$.
Although it is difficult to verify visually,
even enlarged, 
Appendix~\secref{GuardCoords}
verifies that all dark-ray intersections
lie strictly exterior to the triangle.
This demonstrates $g=4n-2$ is achievable for triangles.

\begin{figure}[htbp]
\centering
\includegraphics[width=0.6\textwidth]{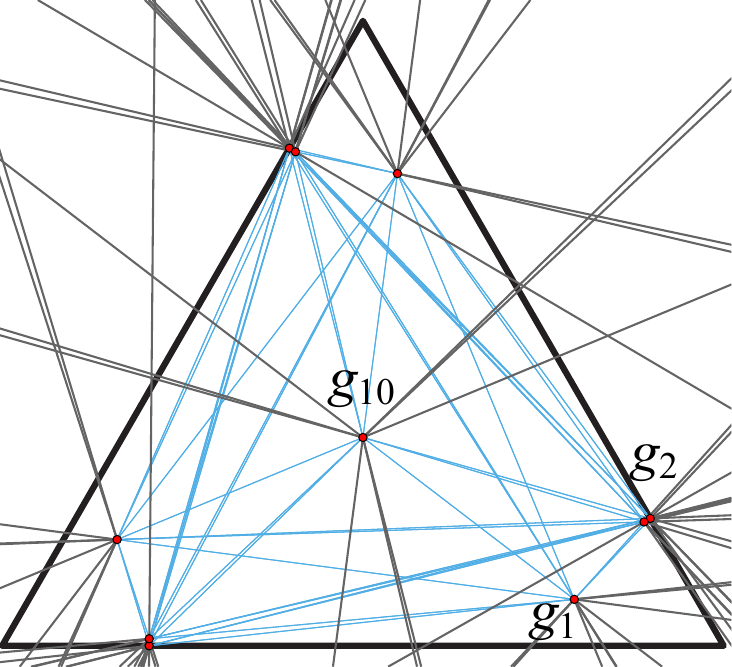}
\\[5ex]
\includegraphics[width=0.6\textwidth]{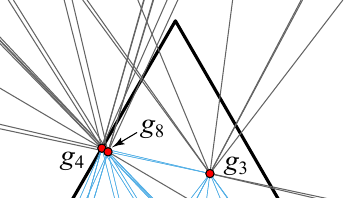}
\caption{$g=10$ guards $9$-covering a triangle.
Apex enlargement below. 
Indexing follows Fig.~\protect\figref{TriangleDarkWedges}.}
\figlab{FreddyTriangle}
\end{figure}
Several features of this construction will repeat for general $n$-gons:
\begin{enumerate}[label={\rm(\arabic*)}]
\squeezelist
\item $n$ guards are on edges of $P$.
\item $2n$ guards are on the hull $\partial C$ (the maximum by Lemma~\lemref{Chull}). 
Recall that $C$ is the convex hull of the guards.
\item Three guards are placed near each vertex,
\item Two of the three guards near a vertex are 
nearly co-located.
\item 
There is one extra guard in each triangle of a triangulation of $P$ (this is $g_{10}$ in Fig.~\figref{FreddyTriangle}).
\end{enumerate}
This construction leads to
three guards near each of $P$'s $n$ vertices,
plus $n-2$ guards in the triangles of a 
triangulation
of $P$, yielding $g=4n-2$.

\paragraph{Idea of the construction in Fig.~\figref{FreddyTriangle}.}
Before turning to the general construction, we first provide intuition for the triangle construction,
illustrated in Fig.~\figref{TriangleDarkWedges}.
The triangle is partitioned into six sectors with $g_{10}$ in the center.
Three guards are placed in the yellow sectors near each vertex, so 
that the dark rays they generate at $g_{10}$
exit through the empty white sectors.
First, two of three guards are placed as illustrated: $g_2,g_4,g_6$ on triangle edges,
and $g_1,g_3,g_5$ slightly inside the adjacent edges.
The final three guards will be placed inside the convex hull of $g_1, \ldots, g_6$, but their locations are tightly constrained.
The guards placed so far
define three dark wedges apexed at 
guards $g_1,g_3,g_5$, where
the wedge apexed at $g_i$ contains 
all the dark rays at $g_i$.
The last three guards $g_7,g_8,g_9$ are placed quite close to the even-index guards $g_2,g_4,g_6$
so that none of their dark rays enter the dark wedges.  For further explanation, see
Section~\secref{GeneralLowerBound}.
The construction works for any triangle:
there are no shape assumptions.

\begin{figure}[htb]
\centering
\includegraphics[width=0.75\textwidth]{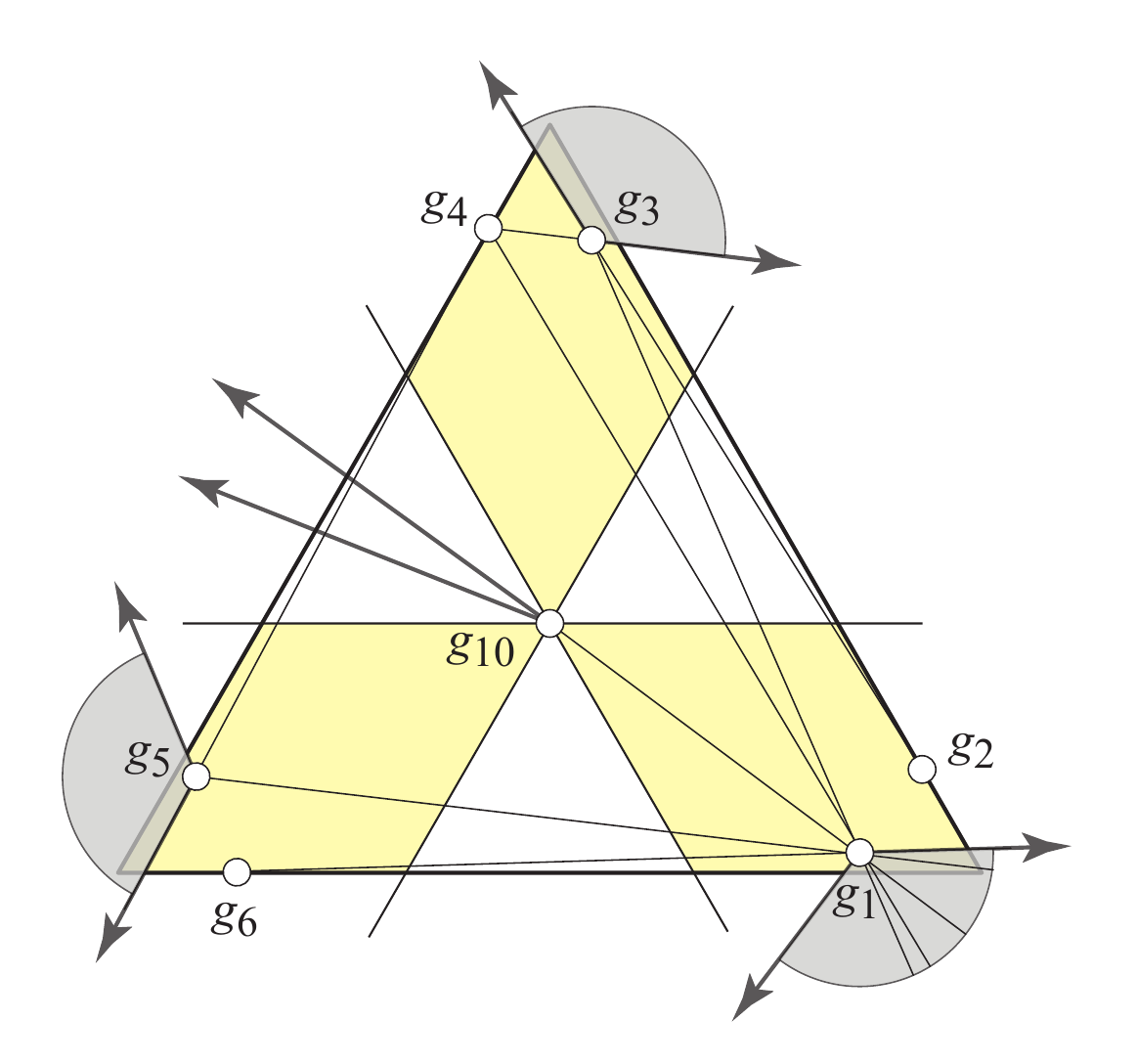}
\caption{Dark rays from $g_{10}$ exit through empty white sectors.
Dark wedges apexed at $g_1,g_3,g_5$ contain the dark rays from all other guards,
illustrated for the $g_1$ wedge.}
\figlab{TriangleDarkWedges}
\end{figure}

\paragraph{Example: Square.}
Placing $4n-2=14$ guards in a square
without any $2$-dark points
follows the same construction as with the triangle:
three guards near each vertex, and
$n-2=2$ 
extra guards $\ell_i$, that we refer to as
``elbow" guards, 
one in each triangle of
a special
triangulation, in this case just a diagonal of the square.
See Fig.~\figref{Square_FS}.
Coordinates may be found in 
Appendix~\secref{GuardCoords}. 

\begin{figure}[htbp]
\centering
\includegraphics[width=0.8\textwidth]{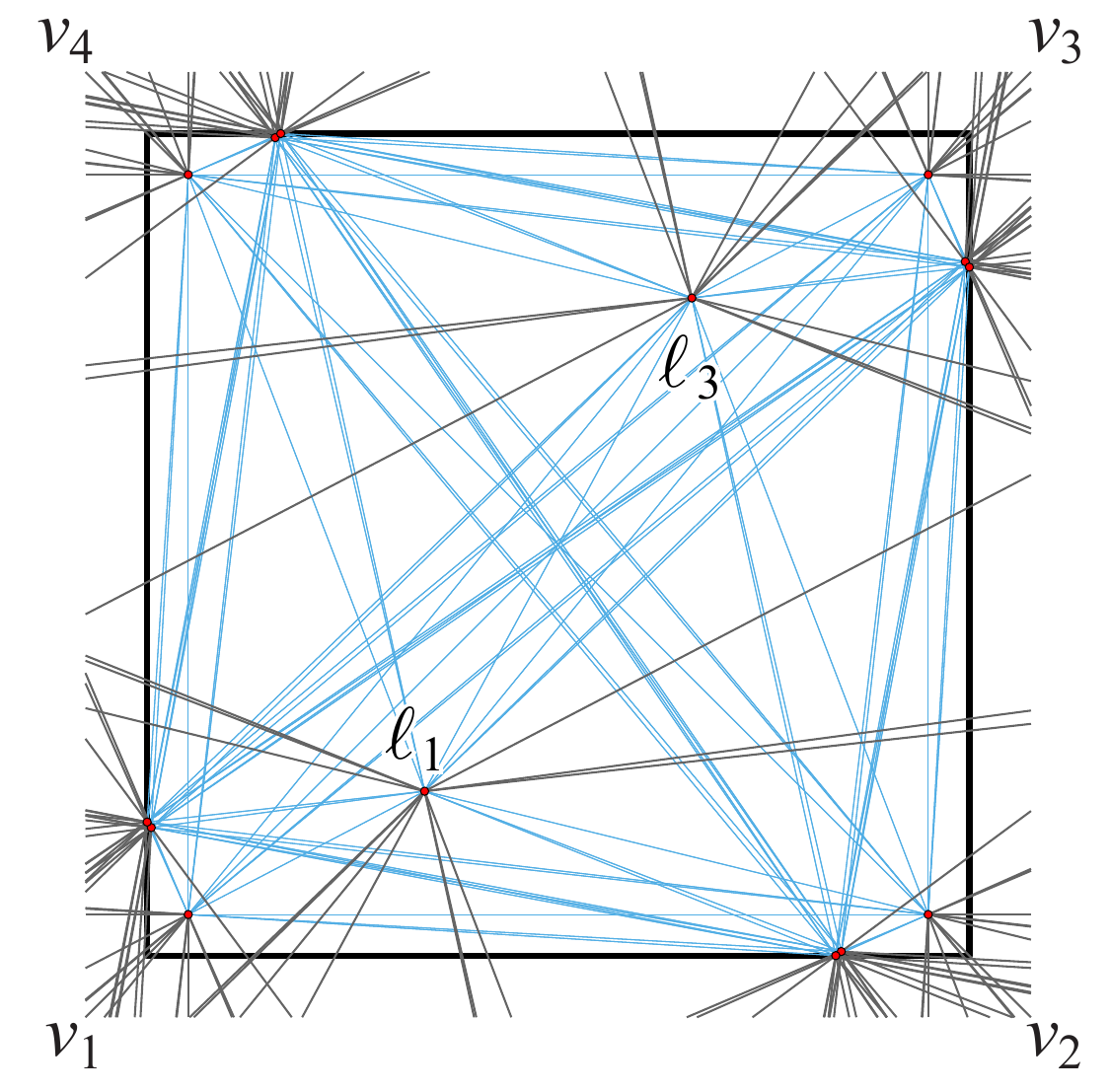}
\caption{
$14$ guards 
$13$-covering a square.
Triangulation diagonal is $v_1 v_3$,
with elbow guards $\ell_1,\ell_3$, one per triangle, plus three vertex guards near each corner.
}
\figlab{Square_FS}
\end{figure}

\noindent
\subsection{General Construction}
\seclab{GeneralLowerBound}

\paragraph{Overall Construction.}
The overall plan of the construction is the same as for a triangle and a square:
$3n$ guards, three near each vertex, plus
one guard per triangle in a triangulation of $P$ of $n-2$ triangles.
The three guards to be placed near $v_i$  
we call
\defn{vertex guards}.
The triangulation is a
\defn{serpentine} triangulation formed by a \defn{zigzag path} that visits all the vertices, as illustrated in Fig.~\figref{ZigZag}. 
The single guard in each triangle will be called an 
\defn{elbow} guard. 

\begin{figure}[htbp]
\centering
\includegraphics[width=0.55\textwidth]{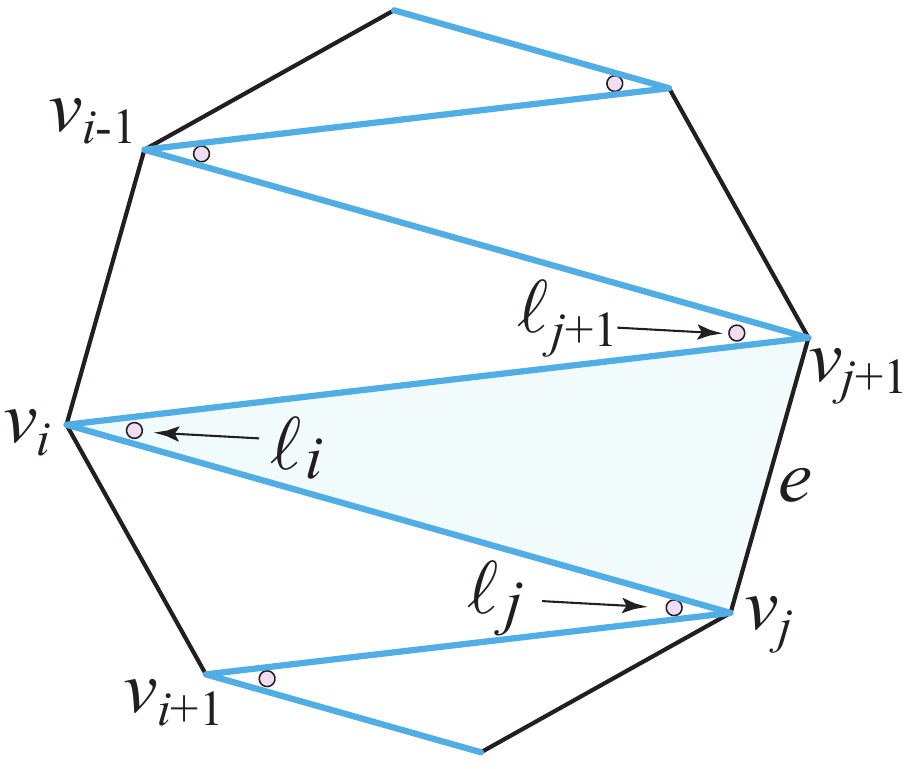}
\caption{Zigzag triangulation and elbow guards $\ell_i$.
}
\figlab{ZigZag}
\end{figure}

\paragraph{Notation.}
We label the vertices in counterclockwise (ccw) order: 
$v_0, \ldots, v_{n-1}$ with index arithmetic modulo $n$.
Thus ``before" means clockwise (cw) and ``after" means ccw.
Let $v_i$ be one  of the $n-2$ internal vertices of the 
zigzag path.  Then $v_i$ is the apex of a triangle $T_i$ bounded by two edges of the zigzag path plus a \defn{base} that is an edge of the polygon. The elbow guard of $T_i$, which we denote $\ell_i$, will be placed close to vertex $v_i$. 
Note that 
the two vertices of $P$ that are the endpoints of the zigzag path have no elbow guard,
and consequently either $\ell_j$  or $\ell_{j+1}$ (or both) might not exist.
For example, in 
the square construction (Fig.~\figref{Square_FS}), 
the zigzag path is $v_4, v_3, v_1, v_2$ so
neither $\ell_2$ nor $\ell_4$ exist.

For ease of notation, we will focus on one triangle with apex $v_i$ and base $v_j v_{j+1}$.
See Fig.~\figref{ElbowWedges}.
In each edge of $P$ we place two ``dividing points'' that are
used to separate wedges of dark rays.
The dividing points adjacent to $v_i$ are labeled
$m_i$ (on the minus (cw) side) and $p_i$ (on the plus (ccw) side).

\paragraph{Dark-ray Wedges.}
The elbow guard $\ell_i$ will be located close to $v_i$, and $v_i$'s three vertex guards will be
even closer to $v_i$. 
We first place the elbow guards and define ``safe regions''
for vertex guards 
so that the dark rays incident to  elbow guards lie in disjoint ``dark ray wedges.''
Exact placement of vertex guards will be described later.

Let $e$ be the base edge of $T_i$, $e=v_j v_{j+1}$.
Refer to Fig.~\figref{ElbowWedges}.
The three portions of $e$ demarcated by $p_j, m_{j+1}$ each are crossed by
wedges of dark rays incident to elbow guards.
The central portion of $e$ is crossed by rays generated by $v_i$'s vertex guards
through $\ell_i$ (blue).
The $v_j p_j$ segment of $e$ is crossed by the rays at $\ell_j$,
generated by all the
vertex guards and elbow guards associated with vertices
ccw from $v_{i+1}$ to $v_{j-1}$,
and symmetrically the $m_{j+1} v_{j+1}$ segment of $e$ is crossed by dark rays
at $\ell_{j+1}$, generated by 
all the
vertex guards and elbow guards associated with vertices 
ccw from $v_{j+2}$ to $v_{i-1}$.

From the viewpoint of $\ell_i$, 
there are three dark wedges emanating from it, one crossing $p_j m_{j+1}$ and two
(shown in pink) crossing $v_i m_i$
and $v_i p_i$, before and after $v_i$.
\begin{figure}[htb]
\centering
\includegraphics[width=0.65\textwidth]{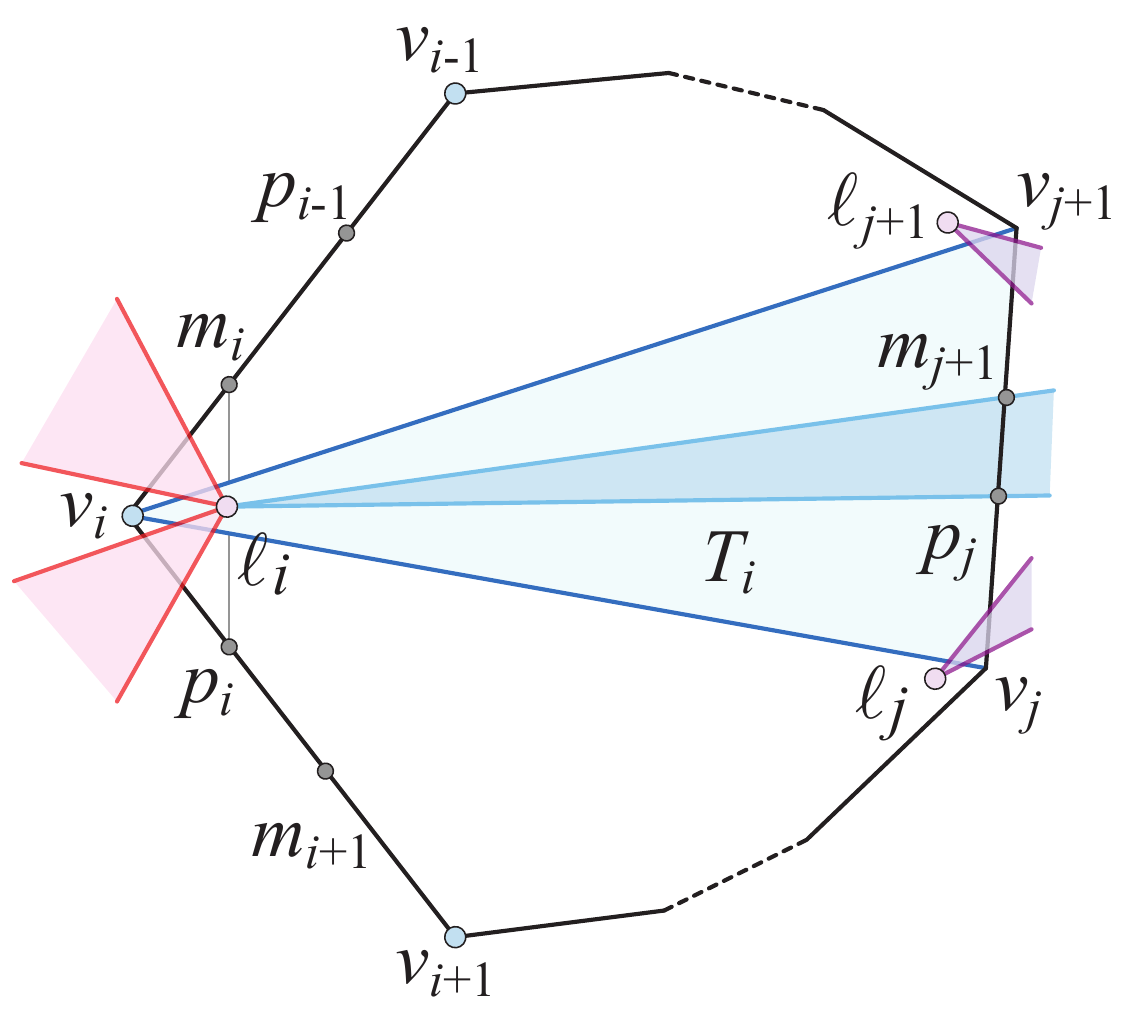}
\caption{
The dark-ray wedges that cross $e=v_j v_{j+1}$ and the dark-ray wedges emanating from~$\ell_i$.
}
\figlab{ElbowWedges}
\end{figure}

\paragraph{Locating $\ell_i$.}
We now describe how to place each $\ell_i$ so that the dark-ray
wedges illustrated in Fig.~\figref{ElbowWedges} indeed contain the claimed rays,
and create a ``safe region" for $v_i$'s vertex guards.

Place $\ell_i$ at the intersection of two lines:
the line $m_i p_i$, and
the line through $v_i$ and the midpoint of $p_j m_{j+1}$.
See Fig.~\figref{Near_v}.

Let $b_i$ be the point where the line through
$p_j$ and $\ell_i$ exits $P$.
Observe that $b_i$ lies in the segment $v_i m_i$.  Our mnemonic is that $b_i$ is just ``before'' $v_i$.
Let $a_i$ be the point where the line through $m_{j+1}$ and $\ell_i$ exits $P$.  Then $a_i$ lies in the segment $v_i p_i$, just ``after'' 
$v_i$.

For a  vertex $v_i$ that has an elbow guard, define 
its \defn{safe region} $R_i$ to be 
the convex quadrilateral $b_i v_i a_i \ell_i$ (pink in Fig.~\figref{Near_v}), which is contained in the triangle $m_i v_i p_i$. 
For a vertex $v_i$ without an elbow guard (the first and last vertices of the zigzag path), its safe region is the triangle $m_i v_i p_i$. 
Observe that the safe regions are pairwise disjoint.

\begin{claim}
If vertex guards for $v_i$ are placed in $R_i$ then the dark rays incident with elbow guards lie in the wedges as specified above
and do not enter the safe regions.
\end{claim}
\begin{proof}
Consider the dark rays incident to $\ell_i$.
Since $v_i$'s vertex guards lie in the wedge 
$a_i \ell_i b_i$, 
they generate dark rays at $\ell_i$ that lie in the complementary wedge $m_{j+1} \ell_i p_j$.
Vertex guards and elbow guards associated with vertices ccw from $v_{i+1}$ to $v_j$ lie in the wedge $p_i \ell_i p_j$ so they generate dark rays at $\ell_i$ that lie in the complementary wedge $m_i \ell_i b_i$ (yellow wedges in Fig.~\figref{Near_v}).
Similarly
vertex and elbow guards associated with vertices ccw from $v_{j+1}$ to $v_{i-1}$ lie in the wedge $m_{j+1} \ell_i m_i$ so they generate dark rays at $\ell_i$ that lie in the complementary wedge $a_i \ell_i p_i$
(green wedges in Fig.~\figref{Near_v}).
\end{proof}

\paragraph{\boldmath Locating $\ell_i$.}
We now describe constraints on locating each $\ell_i$ so that the dark-ray
wedges illustrated in Fig.~\figref{ElbowWedges} indeed contain the claimed rays,
and create a ``safe region" for $v_i$'s vertex guards.

Place $\ell_i$ at the intersection of two lines:
the line through $v_i$ and the midpoint of $p_j m_{j+1}$,
and 
the line $m_i p_i$.
The reason for placing $\ell_i$ on 
the line $m_i p_i$
is to ensure that dark rays from vertex/elbow guards near $v_{i-1}$ through $\ell_i$ 
hit edge $v_i v_{i+1}$ in the segment $v_i p_i$,
and symmetrically dark rays originating near $v_{i+1}$
hit edge $v_i v_{i-1}$ in the segment $v_i m_i$.

Now lines through $\ell_i$ and $p_j$, and through $\ell_i$ and $m_{j+1}$, cross
$P$ just before and just after $v_i$.
See Fig.~\figref{Near_v}.

\paragraph{Claim 1.}
If we remove $v_i$ and consider the convex hull of the remaining vertices plus the elbow guards, then $\ell_i$ is on the convex hull, and the cyclic order of the points around $\ell_i$  
is $v_{i+1}, \ldots, v_j, v_{j+1}, \ldots v_{i-1}$,
where the associated elbow guards lie inside these two ranges.

\paragraph{Safe region for vertex guards.}
One more constraint will establish the safe region for the vertex guards.
Just as $\ell_i$ is located to the $v_i$-side of the line $m_i p_i$,
$\ell_{i-1}$ is located to the $v_{i-1}$-side of the line $m_{i-1} p_{i-1}$.
Therefore
the line through $m_{i-1} \ell_{i-1}$ crosses the edge $v_i v_{i-1}$ before $v_i$,
and symmetrically the line through $p_{i+1} \ell_{i+1}$ crosses the edge $v_i v_{i+1}$ after $v_i$.
See Fig.~\figref{Near_v}.

These lines, together with the constraint from $\ell_i$,
define a convex \defn{safe region} $R_i$,
the pink quadrilateral shown in Fig.~\figref{Near_v},
containing a neighborhood of $v_i$.
These constraints 
ensure that dark rays from the vertex guards in $R_i$ through $\ell_{i-1}$ 
hit the segment $v_{i-1} m_{i-1}$.
Symmetrically, the vertex guards must lie on the $v$-side of the line through 
$\ell_{i+1} p_{i+1}$.

Note that, because $\ell_i$ is to the $v_i$-side of the line $m_i p_i$ through
$v_i$'s adjacent dividing points, the safe region $R_i$
is to that side. This holds for every $i$, and so guarantees that all the safe regions are disjoint.

\paragraph{Claim 2.}
Suppose we place one or more vertex guards in each safe region.  Consider any point $u$ in the safe region of $v_i$ (possibly $u = \ell_i$).  The cyclic order of guards around $u$ (excluding $v_i$’s vertex guards) starts with a vertex guard of $v_{i+1}$ and ends with a vertex guard of $v_{i-1}$.
To prove that a vertex guard, say $w$ at $v_{i-1}$ comes ccw after $\ell_{i-1}$, 
note that the line $w \ell_{i-1}$ hits the boundary of $P$ ccw after 
the base of $v_{i-1}$'s triangle, i.e., after $v_{j+1}$ in Fig.~\figref{Near_v}.
Therefore, the triangle $\triangle u \ell_{i-1} w$ is ccw, establishing the cyclic order of the guards.

Claims~1 and~2 together show that
the dark rays incident with elbow guards lie in the three wedges as specified above and therefore do not intersect inside $P$,
i.e., the dark-ray wedges in Fig.~\figref{ElbowWedges} are accurate as illustrated.

\begin{figure}[htb]
\centering
\includegraphics[width=0.6\textwidth]{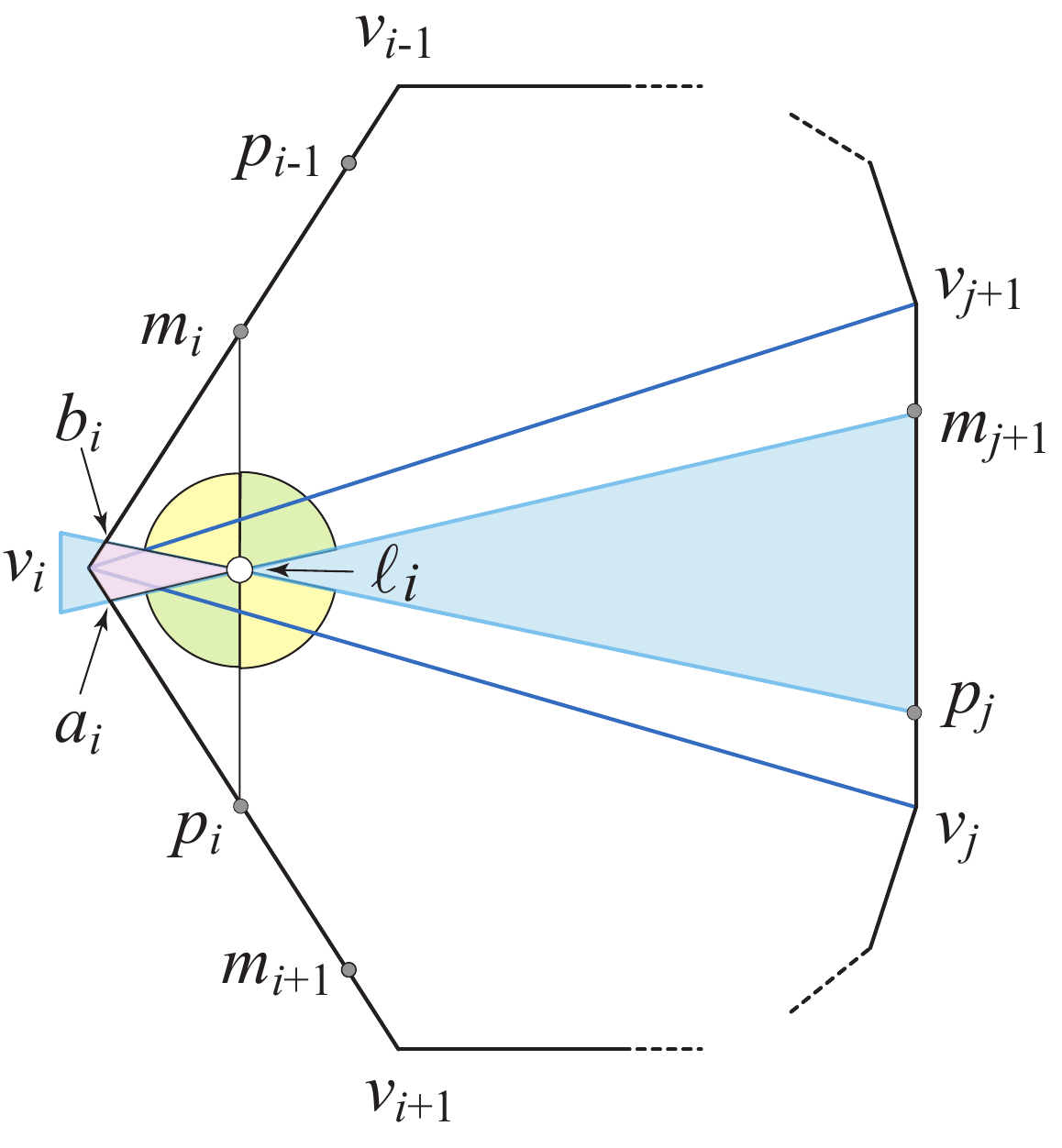}
\caption{
Constraints on locating $\ell_i$, and for locating vertex guards
in a safe region $R_i=b_i v_i a_i \ell_i$.
}
\figlab{Near_v}
\end{figure}

\paragraph{\boldmath Locating $3$ vertex guards.}
Call the three $v_i$ vertex guards $x_i,y_i,z_i$.
We will place them in that order, inside the safe region $R_i$.
$x_i$ will be placed on an edge of $P$,
and $x_i$ and $y_i$ will be on the convex hull $C$ of the guards,
with $z_i$ strictly inside $C$.
\begin{figure}[htb]
\centering
\includegraphics[width=0.8\textwidth]{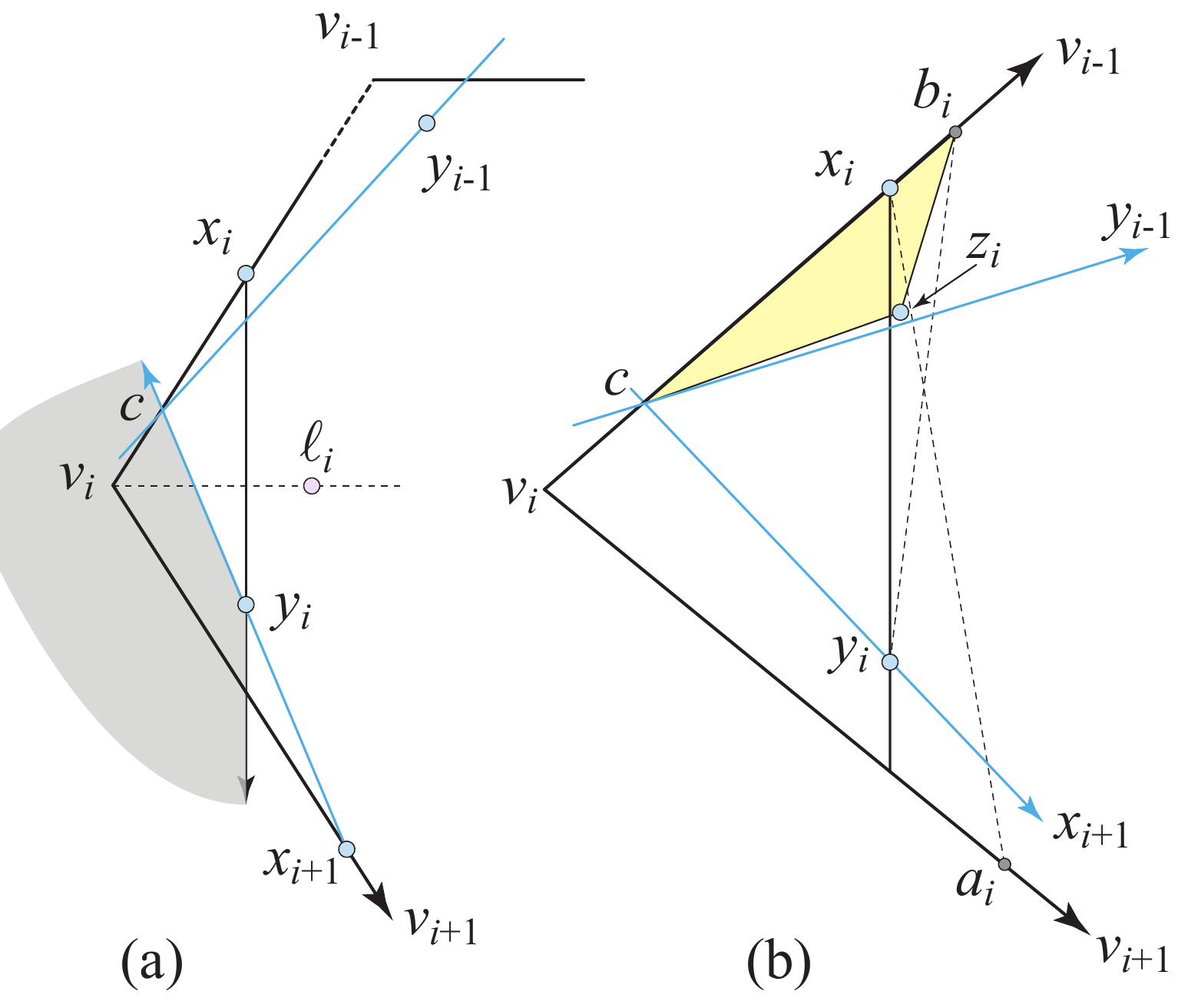}
\caption{(a)~Locating $x_i$ and $y_i$.
Wedge of dark rays apexed at $y_i$ shaded.
(b)~Locating $z_i$ so that dark rays
incident to $z_i$ exit $P$ safely.
}
\figlab{VertexGuards}
\end{figure}

The following construction references $a_i$ and $b_i$ so it applies to the case when $\ell_i$ exists.  But for a vertex $v_i$ without an elbow guard, the same construction works with $m_i$
and $p_i$ serving in place of $b_i$ and $a_i$.

Construct a triangle with apex $v_i$ and two points on $\partial P$
strictly inside the safe region $R_i$.
Place $x_i$ at the corner of this triangle on edge $v_i v_{i-1}$,
and place $y_i$ on the base of the triangle
and on the $p_i$ side of the line $v_i \ell_i$;
see Fig.~\figref{VertexGuards}.
Observe that all the elbow guards are inside the resulting hull $C$.
Because $x_i$ is the only guard on its edge, there are no dark rays incident to $x_i$
inside $P$.
Because $y_i$ lies on $C$ with neighbours $x_i$ and $x_{i+1}$, all
the dark rays incident to $y_i$ lie in 
the complementary wedge bounded by the lines $y_i x_i$ and $y_i x_{i+1}$, and including $v_i$
(gray in Fig.~\figref{VertexGuards}(a)).
Note that no other dark rays intersect this wedge because it lies inside the safe region.

We now place $z_i$.
Let $c$ be the point where 
the line $x_{i+1} y_i$
intersects the edge $v_i v_{i-1}$.
See Fig.~\figref{VertexGuards}(a).

We will ensure that the dark rays incident to $z_i$---except for the one generated by $x_i$---lie in the wedge $c z_i b_i$
(yellow in Fig.~\figref{VertexGuards}(b)).
This implies that 
these rays do not  intersect any other dark rays.

We place $z_i$:
\begin{enumerate}
\squeezelist
    \item inside $C$,
    \item on the $x_i$ side of lines $y_i b_i$ and $y_{i-1} c$,
    \item on the $y_i$ side of line $x_i a_i$.
\end{enumerate}

Observe that 
these constraints
determine a non-empty region for $z_i$.

Conditions 1 and 3 ensure that the dark ray incident to $z_i$ generated by $x_i$ hits the edge $v_i v_{i+1}$ in the segment between $y_i$'s dark wedge and $a_i$, so it intersects no other dark ray.

Conditions 1 and 2 ensure that, if we ignore $x_i$, then $z_i$ lies on the convex hull 
$C_i$
of the remaining guards, with neighbours $y_i$ and $y_{i-1}$.  Therefore the dark rays incident to $z_i$ lie in the complementary 
wedge---apexed at $z_i$ and exterior to 
$C_i$---which 
lies inside the wedge $b_i z_i c$, as required.

We note that, although our construction places guards quite close together, the  coordinates have polynomially-bounded bit complexity, since we used a finite sequence of linear constraints. By contrast, 
irrational coordinates may be required for the conventional art gallery problem in a simple polygon: \cite{Abrahamsen2021}.

Note that at no point
do we rely on the metrical properties of $P$,
so the construction works for all convex polygons.

This completes the proof of Theorem~\ref{thm:no-2-dark}(B).

\section{Simple Polygon}
\seclab{SimplePolygon}
We mentioned in the Introduction that the
variant we are exploring---multiple coverage and guards-blocking-guards---is
not a natural fit for arbitrary simple polygons.
In a convex polygon $P$, each pair of guards sees all of $P$ except
for their dark rays, whereas in an arbitrary polygon,
guard visibility is also blocked by reflex vertices of $\partial P$.

\subsection{Necessity}
The comb example that establishes necessity of
$\lfloor n/3 \rfloor$ guards to cover a simple polygon of $n$
vertices, 
also shows the necessity of
$k \lfloor n/3 \rfloor$ guards to cover to depth $k$---since no guard can see into more than one spike of the comb, each of the $\lfloor n/3 \rfloor$ spikes needs at least $k$ distinct guards.

In fact, if the comb has at least two spikes, then $k \lfloor n/3 \rfloor$ guards also suffice.
The general construction for $k \ge 2$ is
illustrated in Fig.~\figref{Comb} for depth $k=4$ and $n=9$.
Place $k$ guards in a convex arc below each spike of the comb
so that none of the dark rays generated by these guards enters
any spike. Points in a spike are covered to depth $k$
by the $k$ guards below it.
Although many dark rays cross in the base corridor of the comb,
slight vertical staggering of the convex arcs of $k$ guards
ensures that 
no corridor point is at the intersection of three dark rays, which ensures coverage to depth $k$ for $k \ge 2$ and at least two spikes.
\begin{figure}[htbp]
\centering
\includegraphics[width=0.8\textwidth]{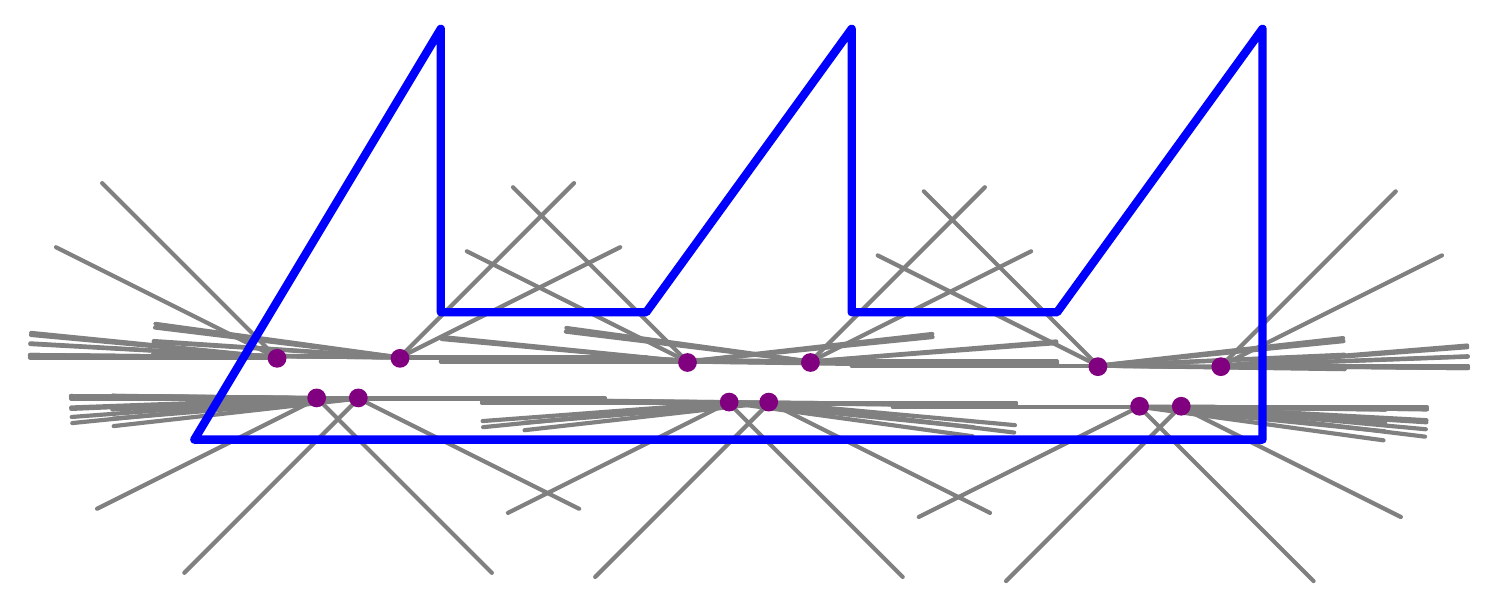}
\caption{$4 \cdot 3 = 12$ guards suffice to $4$-cover the comb of $9$ vertices.}
\figlab{Comb}
\end{figure}

\subsection{Sufficiency}
For sufficiency (beyond the comb example), we have not obtained a tight bound:
To cover a simple polygon $P$ of $n$ vertices to depth $k$,
we show that $g= (k+2) \lfloor n/3 \rfloor$ guards suffice.
Following~\cite{f-spcwt-78} we
first triangulate $P$, $3$-color, and choose the 
smallest color class, 
which has cardinality
at most $\lfloor n/3 \rfloor$.
In Fig.~\figref{SimplePolygon}, say we select color $1$.
If a color-$1$ vertex $v$ is convex, then define a cone
$C$ apexed at $v$ bounded by the edges incident to $v$.
If a color-$1$ vertex $v$ is reflex, then define $C$ to be
the ``anticone" at $v$: the cone apexed at $v$ and bound by the
extensions of the incident edges into the interior.

\begin{figure}[htbp]
\centering
\includegraphics[width=0.5\textwidth]{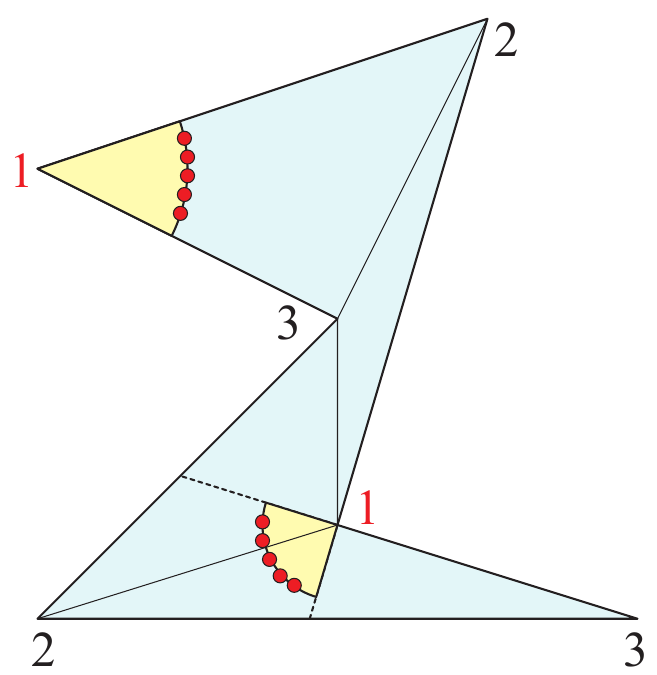}
\caption{Cones at the color-$1$ reflex vertices each contain $k+2$ guards. Here the $5$ guards achieve a $3$-cover.}
\figlab{SimplePolygon}
\end{figure}

To cover $P$ to depth $k$, place $k+2$ guards along a convex arc
near a color-$1$ vertex $v$, and inside $v$'s cone.
In the figure, we aim to $3$-cover and so place $5$ guards in each cone.
Now it is clear that the $k+2$ guards at 
color-$1$ vertex $v$
see into all the triangles incident to $v$.
These guards generate crossing dark rays, but
by Lemma~\ref{lemma:avoid-3-dark} we can perturb the locations of the guards to  
avoid three dark rays meeting in $P$.
The result is coverage to depth two less than the number of
guards at each color-$1$ vertex:

\begin{theorem}
\thmlab{SimplePolygon}
To cover a simple polygon of $n$ vertices to depth $k$,
$g= k \lfloor n/3 \rfloor$ guards are sometimes necessary, and
$g= (k+2) \lfloor n/3 \rfloor$ guards always suffice.
\end{theorem}

\section{\boldmath $10$ Guards in a Wedge}
\seclab{Wedge}
Define a \defn{wedge} as the region of the plane
bounded by two rays from a convex vertex $a$, i.e., a
cone with apex $a$.
The connections between $k$-guarding and dark points
(Observations~\ref{obs:guards-and-rays-1},~\ref{obs:guards-and-rays-2}, and~\ref{obs:guards-and-rays-3})
still hold, and the main issue is the analogue of Theorem~\ref{thm:no-2-dark}---what is the maximum number of guards that can be placed in a wedge without creating $2$-dark points? For a triangle, the bound is $4n-2 = 10$.  In this section we prove that the same bound holds for a wedge.

The analogue of Theorem~\ref{thm:no-2-dark}(A)
is easy: If we could place $11$ guards in a wedge
without $2$-dark points, then we could simply cut off the empty part of the wedge to create a triangle with $11$ guards and no $2$-dark points, a contradiction to 
Theorem~\ref{thm:no-2-dark}(A).

However, 
the analogue of Theorem~\ref{thm:no-2-dark}(B), 
i.e., a placement of $10$ guards without $2$-dark points, does not carry over from our triangle construction, 
because there were dark ray intersections beyond every edge of the triangle.
Nevertheless,
we now show this bound is tight, with the example illustrated
in Figs.~\figref{Wedge10tall} and~\figref{Wedge10}.
We number the guards from bottom to top.
Here is a description of the construction:
\begin{itemize}
    \item $g_1$ is directly below the apex $a$, and far below.
    \item $g_2$ is slightly left of $g_1$, so that the upward dark
    ray at $g_2$ exits the wedge at a particular ``safe" spot
    between $g_7$ and $g_{10}$.
    \item Guard pairs $g_3,g_4$, $g_5,g_6$, $g_7,g_8$
    are symmetrically placed with respect to a vertical line $L$ through $a$.
    \item Guards $g_7,g_8$ are located on the two edges of the wedge.
    \item $g_{10}$ is on $L$ near $a$, while $g_9$ is right of $L$.
    \item There are six guards on the convex hull $C$ of the guards:
    $\{ g_1,g_3,g_7,g_{10},g_8,g_4 \}$.
    \item $g_5,g_6$ are just slightly inside $C$.
\end{itemize}
We provide coordinates for the guards in 
Section~\secref{GuardCoords}, 
and have verified that there are no $2$-dark points in the wedge.

Note that this construction provides an alternative arrangement of guards for a triangle:
Introduce a triangle edge $bc$ below $g_1$, and apply an affine transformation to
$\triangle abc$ to match Fig.~\figref{Wedge10tall}.

There is one difference between the bounds for $k$-guarding a wedge and a triangle: we can $3$-cover a triangle with three guards, one on each edge, but to $3$-cover a wedge, we need an extra guard.
We summarize the implications for $k$-guarding a wedge in this lemma.

\begin{figure}[htbp]
\centering
\includegraphics[height=0.8\textheight]{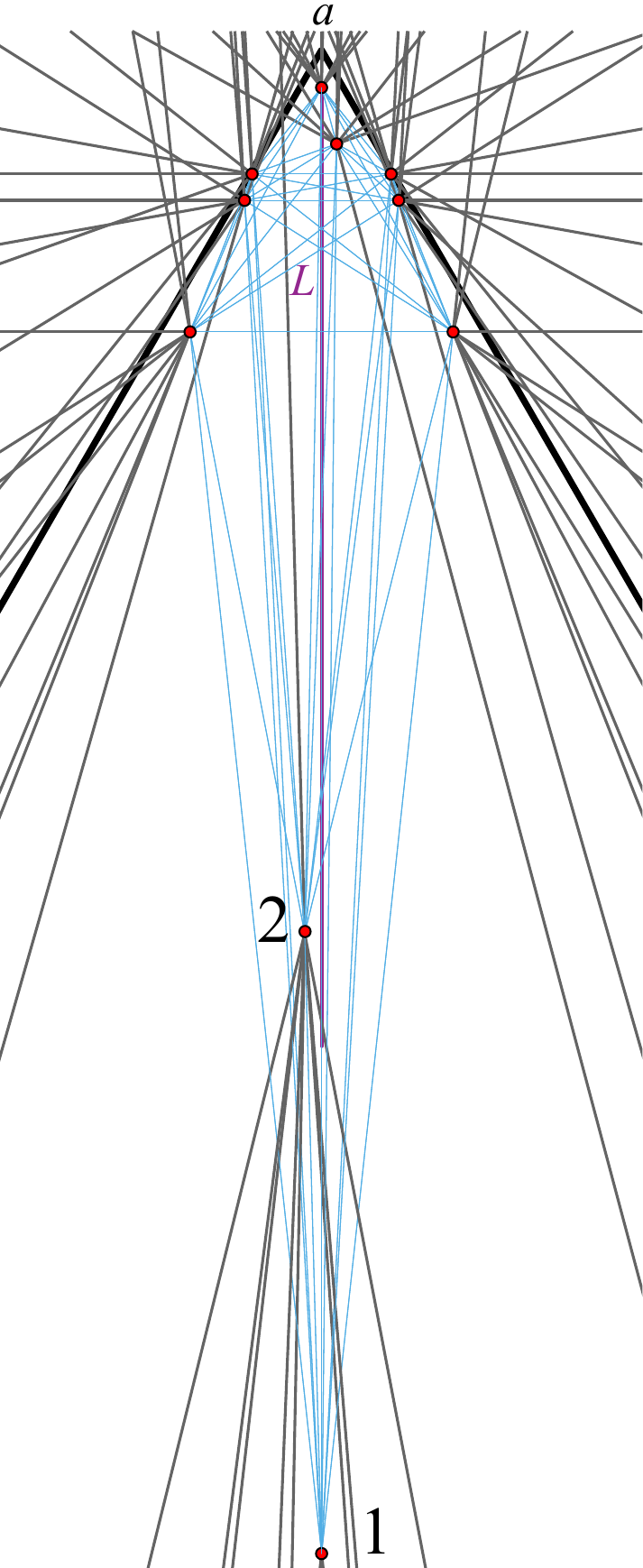}
\caption{Wedge apex $a$, $10$ guards with no $2$-dark points.}
\figlab{Wedge10tall}
\end{figure}
\begin{figure}[htbp]
\centering
\includegraphics[width=0.8\textwidth]{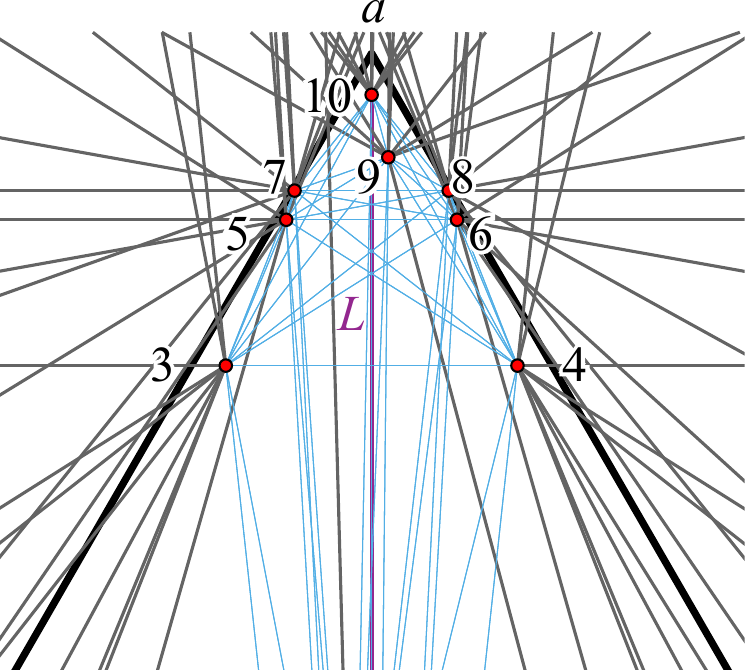}
\caption{Closeup of upper portion of Fig.\protect\figref{Wedge10tall}.}
\figlab{Wedge10}
\end{figure}

\begin{lemma}
\lemlab{Wedge}
Covering a wedge to depth $k$ requires the same number of guards
as it does to cover a triangle to depth $k$,
except that to $3$-cover requires $4$ guards.
In particular, $g=10$ guards can cover to depth $9$.
\end{lemma}
\begin{proof}
If $k \le 2$, a guard at the one vertex,
or one guard on the interior of each edge,
suffices.
However, any placement of 3 guards creates a dark point in the wedge, so for $k \ge 3$, at least $k+1$ guards are needed to $k$-guard.
For $k \le 9$, the configuration just described 
shows that $k+1$ guards suffice---this
covers the middle regime.
For $k \ge 10$, $g=k+2$ guards are needed and suffice,
from 
Observation~\ref{obs:guards-and-rays-3} and Lemma~\ref{lemma:avoid-3-dark}.
\end{proof}

The surprising part of this result is that 10 guards can be placed in a wedge without creating $2$-dark points---despite the fact that our triangle construction (see Fig.~\figref{TriangleDarkWedges}) fails for a wedge because it has $2$-dark points just outside each triangle edge.

\section{Open Problems}
\seclab{Open}
\begin{enumerate}
    \item Investigate bounds or the 
    complexity (NP-hard?) of 
    placing points
    in a simple polygon so that no two dark rays intersect.  
    (As noted in Section~\secref{SimplePolygon}, the connection between this problem and $k$-guarding fails for non-convex polygons.)
    \item  Close the 
    gap in the bounds for a simple polygon
    in Theorem~\thmref{SimplePolygon}.
    \item 
    Can the tight bound for a wedge in Lemma~\lemref{Wedge}
    be generalized to 
    tight bounds for
    unbounded convex polygons with two rays joined by a chain of $n-1$ vertices  and $n-2$ edges?
\end{enumerate}

\acknowledgements{ 
We benefited from suggestions of three referees
and questions at the conference
presentation: \cite{MIT-DarkRays-2023}.}

\appendix
\section{Appendix: Guard Coordinates}
\seclab{GuardCoords}
We include here explicit coordinates for guards in a triangle, a square, and a wedge. In all cases,
Mathematica code has verified that dark-ray intersections are strictly exterior.

Coordinates for $10$ guards in an equilateral triangle,
Fig.~\figref{FreddyTriangle}.
Triangle corners are $(0,200)$, $(\pm 100 \sqrt{3},-100)$.
Guard locations for the other $g_i$ are symmetrical placements following 
Fig.~\figref{TriangleDarkWedges}.
In response to a referee question, we verified that all 
triangle vertices and guard coordinates can be (imperceptibly) adjusted
to be rational.

\begin{center}
\begin{tabular}{|c|rr|}
\hline
    $g_i$ &  $x,$&$y$ \\
\hline\hline
 $5$ & $-102.57,$&$ -96$ \\
 $6$ & $-102.6,$&$ -100$ \\
 $7$ & $-118,$&$ -49$ \\
 $10$ & $0,$&$ 0$ \\
\hline
\end{tabular}
\end{center}

Coordinates for $14$ guards in a square,
Fig.~\figref{Square_FS}.
Square corner coordinates $(\pm 200, \pm 200)$.
Guard locations $g_6,\ldots,g_{14}$ are symmetrical placements of $g_3,g_4,g_5$.

\begin{center}
\begin{tabular}{|c|rr|}
\hline
    $g_i$ &  $x,$ & $y$ \\
\hline\hline
$1$ & $-65,$&$-120$ \\
$2$ & $ 65,$&$ 120$ \\
$3$ & $-180,$&$-180$ \\
$4$ & $-198,$&$ -137.7$ \\
$5$ & $-200,$&$ -135$ \\ 
\hline
\end{tabular}
\end{center}

Coordinates for $10$ guards in a wedge,
Figs.~\figref{Wedge10tall} and~\figref{Wedge10}.
Apex at $(0,200)$, apex angle $\pi/3$.
Guard locations $g_4,g_6,g_8$ are symmetrical
placements of $g_3,g_5,g_7$.

\begin{center}
\begin{tabular}{|c|rr|}
\hline
     $g_i$ &  $x,$&$y$ \\
     \hline\hline
     $1$ & $0,$&$ -600$ \\
     $2$ & $-9,$&$ -270$ \\
     $3$ & $-70,$&$ 50$ \\
     $4$ & $70,$&$ 50$ \\
     $5$ & $-41,$&$ 120$ \\
     $6$ & $41,$&$ 120$ \\
     $7$ & $-38.1,$&$ 134$ \\ 
     $8$ & $38.1,$&$ 134$ \\
     $9$ & $8,$&$ 150$ \\
     $10$ & $0,$&$ 180$ \\
\hline
\end{tabular}
\end{center}

\nocite{*}
\bibliographystyle{abbrvnat}

\label{sec:biblio}

\end{document}